\documentclass[12pt]{article}

\usepackage{amsmath,amssymb}
\usepackage{latexsym}
\usepackage{graphicx}
\usepackage{bm}
\usepackage{hyperref}

\newenvironment{proof}{\begin{trivlist}
                       \item[]{\bf Proof.}
                       \hspace{0cm}}{\hfill $\Box$
                       \end{trivlist}}

\def\Re{\mathop{\rm Re}}

\date{}

\begin{document}


\centerline{}

\centerline{}

\centerline {\Large{\bf Inverse scattering with the data at fixed energy}}

\centerline{}

\centerline{\Large{\bf and fixed incident direction.}}

\centerline{}

\centerline{\bf {A. G. Ramm}} 

\centerline{}

\centerline{Mathematics Department, Kansas State University,}

\centerline{Manhattan, KS 66506-2602, USA}

\centerline{ email: ramm@math.ksu.edu}

\renewcommand{\thefootnote}{\fnsymbol{footnote}}

\newtheorem{Theorem}{\quad Theorem}[section]

\newtheorem{Definition}[Theorem]{\quad Definition}

\newtheorem{Corollary}[Theorem]{\quad Corollary}

\newtheorem{Lemma}[Theorem]{\quad Lemma}

\newtheorem{Example}[Theorem]{\quad Example}

\newtheorem{Remark}[Theorem]{\quad Remark}

\begin{abstract} \noindent Consider the Schr\"{o}dinger operator $-\nabla^2+q$ with a smooth compactly supported potential
$q$, $q=q(x), x \in \mathbf{R}^3$. \\
Let $A(\beta,\alpha, k)$ be the corresponding scattering amplitude,
$k^2$ be the energy, $\alpha \in S^2$ be the incident direction,
$\beta \in S^2$  be the direction of scattered wave, $S^2$ be the
unit sphere in $\mathbf{R}^3$. Assume that $k=k_0 >0$ is fixed, and
$\alpha=\alpha_0$ is fixed. Then the scattering data are $A(\beta)=
A(\beta,\alpha_0, k_0)=A_q(\beta)$ is a function on $S^2$. The following inverse scattering problem is studied:\\
\textit{IP: Given an arbitrary $f \in L^2(S^2)$ and an arbitrary small number $\epsilon >0$, can one find
$q \in C_0^{\infty}(D)$, where $D \in \mathbf{R}^3$ is an arbitrary fixed domain, such that
$||A_q(\beta)-f(\beta)||_{L^2(S^2)} < \epsilon$? }\\
A positive answer to this question is given. A method for
constructing such a $q$ is proposed. There are infinitely many such
$q$, not necessarily real-valued.
\end{abstract}



\section{Introduction}
Consider the scattering problem: \\
Find the solution to the equation
\begin{equation} \label{eq1}
    [\nabla^2+k^2-q(x)]u=0 \quad \text{ in } \mathbf{R}^3,
\end{equation}
such that
\begin{equation} \label{eq2}
    u=e^{ik \alpha\cdot x}+A(\beta,\alpha,k)\frac{e^{ik r}}{r}+o\left(\frac{1}{r}\right), \quad r=|x| \to \infty,
    \quad \frac{x}{r}=\beta,
\end{equation}
where $\alpha \in S^2$ is a given unit vector, $S^2$ is the unit
sphere, $k=const >0$, $k^2$ is the energy, $\alpha$ is the direction
of the incident plane wave $u_0:=e^{ik \alpha \cdot x}$ is the
incident plane wave, $\beta$ is the direction of the scattered wave.
The function
\begin{equation} \label{eq3}
A(\beta,\alpha,k)=A_q(\beta,\alpha,k)
\end{equation}
 is called the scattering
amplitude corresponding to the potential $q(x)$.

If $q(x) \in C_0^{\infty}(\mathbf{R}^3)$ and is a real valued
function, then the scattering problem \eqref{eq1}-\eqref{eq2} has a
unique solution, the scattering solution. There is a large
literature on this topic, see, for example, \cite{A} and and
references therein. The scattering theory has been developed for
much larger classes of potentials, not necessarily smooth and
compactly supported.

We prove existence and uniqueness of the scattering solution
assuming that Im$q\le 0$, see Lemma 2.5 in Section 2.
The inverse scattering problem consists in finding $q(x)$ in a
certain class of potentials from the knowledge of the scattering
data $A(\beta,\alpha,k)$ on some subsets of the set $S^2 \times S^2
\times \mathbf{R}_+$, where $\mathbf{R}_+ = [0, \infty)$. If
$A(\beta,\alpha,k)$ is known everywhere in the above set, then the
inverse scattering problem is easily seen to be uniquely solvable in
the class of $L_0^2(\mathbf{R}^3)$, that is, in the class of
compactly supported square-integrable potentials, and in much larger
class of potentials. If the scattering data is given at a fixed
energy, and $k=k_0>0$ for all $\beta \in S^2$ and all $\alpha \in
S^2$, then uniqueness of the solution to inverse scattering problem
was proved originally in \cite{R228}.  An algorithm for finding
$q(x)$ from the {\it exact} fixed-energy scattering data and from
{\it noisy} fixed-energy scattering data was given in \cite{R425},
where the {\it error estimates} of the proposed solution were also
obtained, see Chapter 5 in  \cite{R470}.

Only recently the uniqueness of the solution to inverse scattering
problem with {\it non-over-determined} data $A(-\beta,\beta, k)$ and
$A(\beta,\alpha_0, k)$  was proved, see \cite{R584}, and
\cite{R589}, \cite{R603}.

The data $A(-\beta,\beta, k), \forall \beta \in S^2$  and all $k>0$
are the {\it back-scattering} data, the $A(\beta,\alpha_0, k),
\forall \beta \in S^2$ and $\forall k >0$ are the {\it fixed
incident direction} data. The scattering data are called
"non-over-determined" if these data depend on the same number of
variables as the unknown potential, that is on three variables in
the above problems.

Note that the data $A(\beta,\alpha, k_0), \forall \alpha, \beta \in
S^2$ and a fixed $k=k_0 >0$ are over-determined: they depend on four
variables while $q$ depends on three variables.

The inverse problem IP with the data 
$A(\beta):=A_q(\beta):=A(\beta,\alpha_0,
k_0)$ is {\it under-determined}: its data depend on two variables. This
problem, in general, does not have a unique solution in  sharp
contrast to the inverse scattering problems mentioned above. The IP
was not studied in the literature.

In this paper the IP is studied. Assume that 
$D\subset \mathbf{R}^3$ is an arbitrary fixed bounded domain.

Let us formulate the inverse problem:

IP: {\it Given an arbitrary $f\in L^2(S^2)$ and an arbitrary small
number $\epsilon>0$, find a $q\in C^\infty_0(D)$ such that  
\begin{equation}\label{eq4}
    ||A(\beta)-f(\beta)||_{L^2(S^2)} < \epsilon.
\end{equation}
}

The IP's formulation differs from the formulation of the inverse
scattering problems discussed earlier:

i) there is no unique solution for the problem we are discussing,

ii) in place of the exact (or noisy) scattering data a function
$f(\beta)$ on $S^2$ is given, which, in general, is not a scattering
amplitude at a fixed $\alpha=\alpha_0$ and a fixed $k=k_0$
corresponding to any potential from $L^2(D)$.

The main results of this paper include: \\
a) A proof of the existence of $q \in C_0^\infty(D)$ such that \eqref{eq4} 
holds;\\
b) A method for finding a potential $q \in C_0^\infty(D)$ for which
\eqref{eq4} holds;\\
c) an analytical formula for a function $h=qu$, where $u$ is the
scattering solution at a fixed $k$ and a fixed $\alpha$, corresponding to 
$q$.

In section 2 we prove that the set $\{A(\beta)\}$ corresponding to
all $q \in C_0^\infty(D)$, is dense in $L^2(S^2)$, and that the set
of functions $\{h\}=\{qu\}$ is dense in $L^2(D)$ when $q$ runs
through all of $L^2(D)$. Here $u=u(x)=u(x,\alpha_0, k_0)$ is the
scattering solution corresponding to the potential $q$, that is, the
solution to the scattering problem \eqref{eq1}-\eqref{eq2} with
$\alpha=\alpha_0 \in S^2$ and $k=k_0 >0$.

In section 3 an analytical formula for $q$ is given. The $q$ computed by this formula
generates $A(\beta)$ satisfying \eqref{eq4}.

 We do not discuss in this paper the relation of our results with the theory
 of creating materials with a desired refraction coefficient, see  \cite{R515}, \cite{R595}.

\section{The density of the set $A(\beta)$ in $L^2(S^2)$}
Let us start by proving Lemma the following lemma.
\begin{Lemma} \label{lm1}
    If the set $\{A(\beta)\}$ is dense in $L^2(S^2)$ when $q$ runs through all of $L^2(D)$, then
    it is dense in $L^2(S^2)$ when $q$ runs through $C_0^\infty(D)$.
\end{Lemma}
\begin{proof}
    The set $C_0^\infty(D)$ is dense in $L^2(D)$ (for example, in $L^2(D)$ norm), and the scattering
    amplitude $A(\beta)$ depends continuously in the norm of $L^2(S^2)$ on $q$, that is,
    \begin{equation} \label{eq5}
        ||A_{q_1}(\beta)-A_{q_2}(\beta)||_{L^2(S^2)} \leq c||q_1-q_2||_{L^2(D)}.
    \end{equation}
    Estimate \eqref{eq5} follows, for example,  from the known lemma of the author (see \cite{R470}, p.262):
    \begin{equation} \label{eq6}
        -4\pi[A_{q_1}(\beta)-A_{q_2}(\beta)]=\int_D [q_1(x)-q_2(x)]u_1(x,\alpha_0, k_0)u_2(x,-\beta, k_0)dx,
    \end{equation}
    and the well-known estimate
    \begin{equation} \label{eq7}
        \sup_{x \in D} |u(x,\alpha, k)| \leq c,
    \end{equation}
    where $c>0$ is a constant depending on the $L^2(D)$ norm of 
   $q$ and uniform with respect
    to  $k \in [a,\infty)$, $a>0$ is a constant and $\alpha \in S^2$. 
Lemma 2.1. is
    proved.
\end{proof}
Thus, in what follows it is sufficient to establish the density of
the set $\{A_q(\beta)\}$ in $L^2(S^2)$ when $q$ runs through
$L^2(D)$.
\begin{Theorem} \label{thm1}
    For any $f \in L^2(S^2)$ and any $\epsilon >0$ there exists a $q \in C_0^\infty(D)$ such that
    estimate \eqref{eq4} holds, where $A(\beta)=A_q(\beta)$ is the
    scattering amplitude corresponding to $q$,
    and $\alpha=\alpha_0 \in S^2, k =k_0 >0$ are fixed.
\end{Theorem}
\begin{proof}
    By Lemma \ref{lm1}, it is sufficient to prove that the set $\{A(\beta)\}$ is dense in $L^2(S^2)$ when
    $q$ runs through all of $L^2(D)$. Assuming the contrary, one finds a function $f \in L^2(S^2)$ which
    is orthogonal in $L^2(S^2)$ to any function $A(\beta)$. It is
    well-known that
    \begin{equation} \label{eq8}
        -4\pi A(\beta)=\int_D e^{-ik\beta\cdot y}h(y)dy, \qquad h(y):=q(y)u(y),
    \end{equation}
    where $u(y)=u(y,\alpha_0, k_0)$ is the scattering solution. Using the
    orthogonality of $f$ to $A(\beta)$, one gets
    \begin{equation} \label{eq9}
        \int_{S^2} f(\beta)\int_D e^{-i k\beta\cdot y}h(y)dy d\beta=0, \quad \forall q \in L^2(D).
    \end{equation}
    We prove later that when $q$ runs through all of $L^2(D)$, the corresponding $h$ runs
    through a set $ \tilde{L}$ dense in $L^2(D)$.
    Taking this for granted, one can replace in \eqref{eq9} the 
   expression $\forall q \in L^2(D)$ by
    $\forall h \in L^2(D)$, and rewrite \eqref{eq9} as follows:
    \begin{equation} \label{eq10}
        \int_D h(y) \int_{S^2} f(\beta) e^{-ik\beta\cdot y}d\beta dy=0, 
\quad \forall h \in \tilde{L}.
    \end{equation}
    This implies
    \begin{equation} \label{eq11}
        \int_{S^2} f(\beta) e^{-ik \beta\cdot y}d\beta =0, \quad \forall y \in L^2(D),
    \end{equation}
    where $k=k_0$ is fixed. In what follows, we write everywhere $k$ for $k_0$ and $\alpha$ for $\alpha_0$.
    The integral in \eqref{eq11} can be considered as the Fourier transform of a compactly supported distribution
    \begin{equation} \label{eq12}
        g(\lambda,\beta):=g(\xi):=f\left(\frac{\xi}{|\xi|}\right)\frac{\delta(|\xi|-k)}{|\xi|^2},
    \end{equation}
    where $\xi \in \mathbf{R}^3, \lambda=|\xi|, \beta=\frac{\xi}{\lambda},$ and
    \begin{equation} \label{eq13}
        \tilde{g}(y):=\int_{\mathbf{R}^3} g(\xi) e^{-i\xi\cdot y}d\xi=
        \int_0^\infty\lambda^2d\lambda \int_{S^2}e^{-i\lambda \beta \cdot y} g(\lambda,\beta)d\beta.
    \end{equation}
    Since distribution \eqref{eq12} is supported on the sphere $|\xi|=k$, which is a compact set
    in $\mathbf{R}^3$, its Fourier transform is an entire function of $y$.
    This function vanishes in an open in $\mathbf{R}^3$ set $D$ by \eqref{eq11}. Therefore, it
    vanishes everywhere in $\mathbf{R}^3$. By the injectivity of the Fourier transform one concludes
    that $f(\beta)=0$. Therefore the assumption that the set $\{A(\beta)\}$ is not dense in $L^2(S^2)$
    is false. Theorem \ref{thm1} is proved under the assumption that the set $\{h\}$ is dense in $L^2(D)$ when
    $q$ runs through all of $L^2(D)$.
In Theorem \ref{thm2}, see below,  this density statement is proved.
Thus, one can consider Theorem \ref{thm1} proved.
\end{proof}
\begin{Remark}
    If one defines $A(\beta)=-\frac{1}{4\pi}\int_D e^{-ik \beta\cdot y} h(y) dy$ and assumes that $\{h\}$
    runs through a dense subset of $L^2(D)$, then the corresponding set $\{A(\beta)\}$ is dense in
    $L^2(S^2)$, as follows from our proof of Theorem \ref{thm1}.
\end{Remark}
\begin{Theorem} \label{thm2}
    The set $\{h(x)\}$ is dense in $L^2(D)$ when $q(x)$ runs through a dense subset of $L^2(D)$,
    where $h(x)=q(x)u(x)$, and $u(x)$ is the scattering solution.
\end{Theorem}
\begin{proof}
    If $u$ is the scattering solution, then
    \begin{align}
        &u(x)=u_0(x)-\int_D g(x,y)h(y)dy, \quad u_0(x):=e^{ik \alpha\cdot x},\quad h=qu, \label{eq14} \\
        &g(x,y)=\frac{e^{i k|x-y|}}{4\pi|x-y|}. \label{eq15}
    \end{align}
    Define
    \begin{equation} \label{eq16}
            q(x)=\frac{h(x)}{u(x)}=\frac{h(x)}{u_0(x)-\int_D g(x,y) h(y) dy}.
    \end{equation}
    If the function $q(x)$, defined in \eqref{eq16}, belongs to $L^2(D)$, then the function $u(x)$,
    defined in \eqref{eq14}, is the scattering solution, corresponding to $q \in L^2(D)$,
    defined by formula \eqref{eq16}. Uniqueness of the scattering 
solution is guaranteed if Im$q \le 0$ by
    the following Lemma \ref{lm2}.
    \begin{Lemma} \label{lm2}
        Assume that Im$q \leq 0$ and $q \in L^2(D), q=0$ in 
$D:=\mathbf{R}^3 \setminus D$. Then there
        exists a unique scattering solution, that is, the solution to problem \eqref{eq1}-\eqref{eq2}.
        This solution is also the unique solution to equation \eqref{eq14} if $h=qu$.
    \end{Lemma}
    \begin{proof}
        It is sufficient to prove uniqueness of the scattering solution. Indeed, the scattering solution
        solves a Fredholm-type Lippmann-Schwinger equation,
$$u(x)=u_0(x)-\int_D g(x,y)q(y)u(y)dy,$$
        and the uniqueness of the solution to
        this equation implies the existence of this solution by the Fredholm alternative.

        Suppose that there are two scattering solutions, $u_1$ and $u_2$, that is, solution to problem
        \eqref{eq1}-\eqref{eq2} (with $\alpha=\alpha_0$ and $k=k_0 >0$). Then the function
        $v:=u_1-u_2$ solves equation \eqref{eq1} and satisfies the radiation condition at infinity:
        \begin{align}
            &[\nabla^2+k^2-q(x)]v=0 \quad \text{ in } \mathbf{R}^3, \label{eq17} \\
            &\frac{\partial v}{\partial r}-ik v=o\left(\frac{1}{r}\right), \quad r:=|x| \to \infty. \label{eq18}
        \end{align}
        Multiply equation \eqref{eq17} by $\bar{v}$, the bar stands for complex conjugate, and the
        complex conjugate of \eqref{eq17} by $v$ and subtract from the first equation the second. The result is
        \begin{equation} \label{eq19}
            \bar{v}(\nabla^2-k^2)v - v(\nabla^2-k^2)\bar{v}-(q-\bar{q})|v|^2=0.
        \end{equation}
        Integrate \eqref{eq19} over a ball $B_R$ of large radius $R$, centered at the origin, and use
        the Green's formula to get
        \begin{equation} \label{eq20}
            \int_{|x|=R} \left(\bar{v}\frac{\partial v}{\partial r}-v\frac{\partial \bar{v}}
            {\partial r}\right)ds - 2i\int_{B_R} \text{Im}q(x)|v(x)|^2dx=0.
        \end{equation}
        Using the radiation condition \eqref{eq18} one rewrites \eqref{eq20} as
        \begin{equation} \label{eq21}
            2ik\int_{|x|=R} |v|^2 ds - 2i \int_{B_R} \text{Im}q|v|^2 dx +o(1)=0,
        \end{equation}
        where $o(1) \to 0$ as $R \to \infty$. Thus, if Im$q \leq 0$ relation \eqref{eq21} implies that
        \begin{equation} \label{eq22}
            \lim_{R\to \infty} \int_{|x|=R} |v|^2 ds=0.
        \end{equation}
        The radiation condition \eqref{eq18} and equation \eqref{eq17} with a compactly supported
        $q \in L^2(D)$, implies that $v=0$ in $\mathbf{R}^3$, see Lemma \ref{lm1} on p.25 in \cite{R190}.

        The scattering solution solves equation \eqref{eq14}, because it satisfies the equation
        \begin{equation} \label{eq23}
            u=u_0 -\int_D g(x,y) q(y) u(y) dy,
        \end{equation}
        so that with the notation $h:=qu$ one gets from \eqref{eq23} equation 
      \eqref{eq14}. 

Conversely,
        assume that $u$ is defined by equation \eqref{eq14} and $h$ in \eqref{eq14} is equal to $qu$,
        where $q$ is defined by formula \eqref{eq16} and $q \in L^2(D)$ . Then this $u$ solves
        equation \eqref{eq23}, and, therefore, it satisfies equations 
\eqref{eq1} and \eqref{eq2}. To
        check this, apply the operator $\nabla^2+k^2$ to equation \eqref{eq14} and use the known formula
        $(\nabla^2+k^2)g=-\delta(x-y)$.

        The result is $(\nabla^2+k^2)u=h=qu$, so equation \eqref{eq1} holds. The radiation condition holds
        because $q$ is compactly supported and $g$ satisfies the radiation condition. Lemma \ref{lm2} is proved.
    \end{proof}
    Although the scattering solution, in general, is not unique if Im$q>0$ and may not exist in this case,
    it does exist, even if Im$q>0$, if $q$, defined by formula \eqref{eq16}, belongs to $L^2(D)$. Indeed,
    then $qu=h$ and $$u(x)=u_0(x)-\int_D g(x,y)q(y)u(y)dy,$$ so that 
$u:=u_0(x)-\int_D g(x,y)h(y)dy$
    is the scattering solution.

    Let us assume now that the function \eqref{eq16} does not belong to
    $L^2(D)$. Since $h \in L^2(D)$
    the function $u(x) \in H_{loc}^2(\mathbf{R}^3)$, where
    $H_{loc}^2(\mathbf{R}^3)$ is the Sobolev space.
    
Changing $h$ slightly  one may assume that $h$ is a smooth
    bounded function in $D$. Such a change leads to a small
    change of the scattering data $A(\beta)$  in $L^2(S^2)$ norm.
    Thus, we will assume below that $h$ is bounded in $D$ in
    absolute value.
    The function \eqref{eq16} does not belong to $L^2(D)$ if and only if
the denominator
    in \eqref{eq16} has zeros. Let
    $$N:=N(u):=\{x: u(x)=0, x \in D\}, \quad N_\delta:=\{x: |u(x)| < \delta, x \in D \},$$
    where $\delta >0$ is a small number, and let $D_{\delta}:=D \setminus N_{\delta}$.

    The idea of the argument below is to show that the set $N$ is generically a line in $\mathbf{R}^3$,
    and that there exists a function
    \begin{equation} \label{eq24}
        h_{\delta}(x)=\left\{
                        \begin{array}{ll}
                            h(x) &\text{ in } D_{\delta}, \\
                            0    &\text{ in } N_{\delta},
                        \end{array}
                    \right.
    \end{equation}
    such that
    \begin{equation} \label{eq25}
        ||h_\delta - h||_{L^2(D)} \leq c\delta^2.
    \end{equation}
    Moreover, the function $h_\delta\in L^\infty(D)$. This function can be made
    smooth by an approximation by a $C^\infty_0(D)$-function. The corresponding
    smooth $h_\delta$ generates $A_\delta(\beta)$ which differs slightly from the original $f(\beta)$.
    The  corresponding $u_\delta(x)$ can be defined as
    follows:
    \begin{equation} \label{eq26}
        u_\delta(x):= u_0(x)-\int_D g(x,y)h_\delta(y)dy,
    \end{equation}
    where
    \begin{equation} \label{eq27}
        q_{\delta}:=\left\{
                        \begin{array}{ll}
                            \frac{h_\delta}{u_\delta} &\text{ in } D_{\delta}, \\
                            0                         &\text{ in } N_{\delta},
                        \end{array}
                    \right.
        \qquad q_\delta \in L^\infty(D).
    \end{equation}
    Consequently, we will prove that a small change of $h$  may be arranged in such a way that the
    corresponding change of $q$ leads to a potential $q_\delta$ which 
belongs to $L^\infty(D)$.
    \begin{Lemma} \label{lm3}
        The set $N$ is a line in $\mathbf{R}^3$.
    \end{Lemma}
    \begin{proof}
        Let $u=u_1+iu_2$, where $u_1=\Re u$ and $u_2=\Im u$. Then the set $N$ is defined by two equations
        in $\mathbf{R}^3$:
        \begin{equation} \label{eq28}
            u_1(x)=0, \quad u_2(x)=0, \quad x \in \mathbf{R}^3.
        \end{equation}
        The functions $u_j \in H_{loc}^2(\mathbf{R}^3)$, $j=1,2$, because $u \in H_{loc}^2(\mathbf{R}^3)$.
        Therefore, each of the two equations in \eqref{eq28} is an equation of a surface.

        The two simultaneous equations \eqref{eq28} generically describe a 
line $\ell$ in $\mathbf{R}^3$.
        By a small perturbation
        of $h$ one may ensure that the line $\ell:=\{x: u_1(x)=0, u_1(x)=0, x \in D\}$ is  smooth in
        $D$ and the vectors $\nabla u_j, j=1,2,$ on $\ell$ are linearly 
independent.

        Lemma \ref{lm3} is proved.
    \end{proof}
    \begin{Lemma} \label{lm4}
        There exists a function \eqref{eq24} such that \eqref{eq25} holds.
    \end{Lemma}
    \begin{proof}
        Consider a tubular neighborhood of the line $\ell$ in $D$. This neighborhood is described by
        the inequality $\rho(x,\ell)\leq \delta$, where $\rho(x,\ell)$  is the distance between $x$ and $\ell$.
        Choose the origin $O$ on $\ell$ and lets the coordinates $s_j, j=1,2$, in the plane orthogonal to $\ell$,
        be directed along the vectors $\nabla u_j|_l$ respectively, while the third coordinate $s_3$
        be directed along the tangent line to $\ell$. The Jacobian $J$ of the transformation
        $(x_1, x_2, x_3) \to (s_1, s_2, s_3)$ is non-singular, 
  $|J|+|J^{-1}| \leq c$, because
        $\nabla u_j, j=1,2$ are linearly independent.

        Define $h_\delta$ by formula \eqref{eq24}, $u_\delta$ by formula \eqref{eq26},
        and $q_\delta$ by formula \eqref{eq27}.

        Note that
        \begin{equation} \label{eq29}
            u_\delta(x)=u(x)+\int_D  g(x,y)[h(y)-h_\delta(y)]dy.
        \end{equation}
        Since $h_\delta=h$ in $D_\delta$ and $h_\delta=0$ in $N_\delta$, one
        obtains inequality \eqref{eq25}:
      $$||h_\delta-h||_{L^2(D)}=||h||_{L^2(N_\delta)}\le
      c\delta^2,$$
      because $h$ is bounded and $\delta^2$ is proportional to the
      area of the cross section of the tubular neighborhood.
      Furthermore,
        \begin{equation} \label{eq30}
            |u_\delta(x)| \geq |u(x)|-c\int_{N_\delta} \frac{dy}{4\pi|x-y|}, \quad c=\max_{x \in N_\delta}|h(x)|.
        \end{equation}
        Denote
        \begin{equation} \label{eq31}
            I(\delta):= \sup_{x \in D_\delta}\int_{N_\delta} \frac{dy}{4\pi|x-y|}.
        \end{equation}
        By construction
        \begin{equation} \label{eq32}
            |u(x)| \geq \delta \quad \text{ if } x\in D_\delta.
        \end{equation}
        Therefore, inequality \eqref{eq30} implies
        \begin{equation} \label{eq33}
            |u_\delta(x)| \geq \delta - I(\delta) \quad \forall x\in D_\delta.
        \end{equation}
        Let us estimate $I_\delta$ as $\delta \to 0$ with the aim to prove that
        \begin{equation} \label{eq34}
            \inf_{x\in D_\delta}|u_\delta(x)|  \geq b(\delta) >0, \qquad \lim_{\delta \to 0} \frac{b(\delta)}{\delta}=1.
        \end{equation}
        So, for sufficiently small
        $\delta >0$ one has $b(\delta) \geq \frac{\delta}{2}$.

        If \eqref{eq34} holds, then $|u_\delta(x)|$ is strictly positive in $D_\delta$ and, therefore,
        $q_\delta:=\frac{h_\delta}{u_\delta}$ is a bounded function in $D_\delta$.

        Since $h_\delta=h$ in $D_\delta$, one obtains
        \begin{equation} \label{eq35}
            ||q_\delta-q||_{L^2(D_\delta)}=\left|\left|\frac{h_\delta}{u_\delta}-\frac{h}{u}\right|\right|_{L^2(D_\delta)}
            \leq ||h||_{C(D_\delta)}\left|\left|\frac{u-u_\delta}{u_\delta u}\right|\right|_{L^2(D_\delta)}
            \leq c\delta\ln{\frac{1}{\delta}}.
        \end{equation}
        Here the following inequalities  for $x\in D_\delta$ were used
        \begin{align}
            &||h||_{C(D_\delta)} \leq c, \quad |u_\delta u| \geq \frac{\delta^2}{2} \quad\text{ in } D_\delta,\nonumber \\
            &|u_\delta(x)- u(x)|\leq  \int_{N_\delta}|g(x,y)h(y)| dy \leq c\delta \int_{N_\delta} \frac{dy}{4\pi|x-y|}
            \leq c\delta I(\delta). \label{eq36}
        \end{align}
        Here  and throughout by $c>0$ various constants, independent of $\delta$, are denoted.

        To estimate $I(\delta)$ one argues as follows:
        \begin{equation} \label{eq37}
            I(\delta) \leq \frac{1}{4\pi} 
\int_{N_\delta}\frac{dy}{|y|}\leq c\int_0^{c\delta} \rho d\rho
            \int_0^1\frac{ds_3}{\sqrt{\rho^2+s_3^2}}\leq c \delta^2 \ln{\frac{1}{\delta}},
        \end{equation}
        where the unit $1$ is a finite coordinate along the $s_3$ axis 
and  we have used the following estimate:
        \begin{equation} \label{eq38}
            \int_0^1 \frac{ds_3}{\sqrt{\rho^2+s_3^2}}=\ln\left(\left.s_3+\sqrt{\rho^2+s_3^2}\right)\right|_0^1
            \leq c\ln{\frac{1}{\rho}}, \quad \rho \to 0,
        \end{equation}
        where $c>1$ is a constant. Estimates \eqref{eq33} and \eqref{eq38} imply
        inequality \eqref{eq34}. Thus, the existence of a function
         \eqref{eq24} is proved, and
        \begin{equation} \label{eq39}
            ||h_\delta-h||_{L^2(D)}=||h||_{L^2(N_\delta)} \leq c\int_{N_\delta}dx \leq c\delta^2.
        \end{equation}
        Lemma \ref{lm4} is proved.
    \end{proof}
    From Lemma \ref{lm4} the conclusion of Theorem \ref{thm2} follows. Indeed, if $h \in  L^2(D)$ is an
    arbitrary function and $q$, defined by formula \eqref{eq16}, belongs to $L^2(D)$, then, as was proved
    above, the $u$, defined by formula \eqref{eq14}, is the scattering solution. This scattering solution is
    unique if Im$q \leq 0$ by Lemma \ref{lm2}.
    If $q$, defined by formula \eqref{eq16}, does not belong to $L^2(D)$, then there is a bounded
    function $q_\delta(x)$, approximating $q$, see \eqref{eq35}, such that the corresponding
    $h_\delta$ approximates $h(x)$ well, see \eqref{eq39}, and the corresponding $u_\delta(x)$ is
    the scattering solution corresponding to $q_\delta(x)$.

    Theorem \ref{thm2} is proved.
\end{proof}
\section{Formulas for solving inverse scattering problem with fixed $\alpha$ and $k>0$}
The inverse problem (IP) was formulated in the Introduction. Given
$\epsilon>0$ and an arbitrary $f(\beta) \in L^2(S^2)$ we first find
$h(x) \in L^2(D)$ such that
\begin{equation} \label{eq40}
    \left|\left|f(\beta)+\frac{1}{4\pi}\int_D e^{-ik\beta\cdot y}h(y)dy\right|\right|_{L^2(S^2)} \leq \epsilon.
\end{equation}
This can be done (non-uniquely !) in many ways. Let us describe one
of the ways. Without loss of generalities assume that $D=B=B_R$ is a
ball of radius $R$ centered at the origin. Expand the plane wave
$e^{-ik\beta\cdot y}$ and $h(y)$ into the spherical harmonics
series:
\begin{align}
    &e^{-ik\beta\cdot y}=\sum_{l=0}^\infty 4\pi i^l j_l(k r)Y_l(-\beta)\overline{Y_l(y^0)},
    \quad r=|y|, \quad y^0:=\frac{y}{r}, \label{eq41} \\
    &j_l(r):=\left(\frac{\pi}{2r}\right)^{1/2}J_{l+\frac{1}{2}}(r), \quad \frac{4\pi}{2l+1}\sum_{m=-l}^l
    Y_{lm}(x^0)\overline{Y_{lm}(y^0)}=P_l(x^0\cdot y^0), \label{eq42}
\end{align}
$J_{l+\frac{1}{2}}(r)$ is the Bessel function regular at the origin,
$Y_l(\alpha)$ are the spherical harmonics:
\begin{align}
    &Y_l(\alpha)=Y_{l,m}(\alpha)=\frac{(-1)^{\frac{m+|m|}{2}l}i^l}{\sqrt{4\pi}}\left[\frac{(2l+1)(l-|m|)!}
    {(l+|m|)!}\right] ^{1/2} e^{im\varphi} P_{l,m}(\cos\theta), \nonumber \\
    &\qquad -l\leq m\leq l, \label{eq43} \\
    &P_{l,m}(\cos\theta)=(\sin\theta)^m\frac{d^m P_{l}(\cos\theta)}{(d\cos\theta)^m},
    \quad P_l(t)=\frac{1}{2^l l!}\frac{d^l (t^2-1)^l}{dt^l}, \label{eq44}
\end{align}
$t=\cos\theta$, the unit vector $\alpha$ is described by the
spherical coordinates $(\theta,\varphi), 0 \leq \varphi < 2\pi, 0
\leq \theta \leq \pi, -l\leq m\leq l$, one has
\begin{equation} \label{eq45}
    Y_{l,m}(-\alpha)=(-1)^l Y_{l,m}(\alpha), \quad\overline{Y_{l,m}(\alpha)}=(-1)^{l+m} Y_{l,-m}(\alpha),
\end{equation}
where the overline stands for complex conjugate. The summation in
\eqref{eq41} and below is understood as $\sum_{l=0}^\infty
\sum_{m=-l}^l$.

Let
\begin{equation} \label{eq46}
    h(y)=\sum_{l=0}^\infty h_l(r) Y_l(y^0).
\end{equation}
It is well-known that
\begin{equation} \label{eq47}
    (Y_{l,m}Y_{l',m'})_{L^2(S^2)}=\delta_{ll'}\delta_{mm'},
\end{equation}
where $\delta_{ll'}$ is the Kronecker delta.

Let
\begin{equation} \label{eq48}
    f_L(\beta):=\sum_{l=0}^L f_l Y_l(\beta),
\end{equation}
where $f_l:=(f,Y_l)_{L^2(S^2)}$ are the Fourier coefficients of $f$.
For sufficiently large $L$ one has
\begin{equation} \label{eq49}
    ||f-f_L||_{L^2(S^2)} < \epsilon/2.
\end{equation}
Thus, if
\begin{equation} \label{eq50}
    \left|\left|f_L(\beta)+\frac{1}{4\pi}\int_B e^{-ik\beta\cdot y}h(y)dy\right|\right|_{L^2(S^2)} < \epsilon/2,
\end{equation}
then inequality \eqref{eq40} holds. Therefore, practically it is
sufficient to find $h$ satisfying inequality \eqref{eq50}.
Substitute \eqref{eq46} and \eqref{eq41} into the equation
\begin{equation} \label{eq51}
    \int_B e^{-ik\beta\cdot y}h(y)dy=-4\pi f_L(\beta),
\end{equation}
and use \eqref{eq47} and \eqref{eq45} to get
\begin{equation} \label{eq52}
    4\pi (-i)^l \int_0^R  r^2j_l(kr) h_l(r) dr= -4\pi f_l^{(L)}, \quad 0 \leq l \leq L.
\end{equation}
 Equation \eqref{eq52} can be written as
\begin{equation} \label{eq53}
    (-i)^{l+2} \int_0^R  r^2 j_l(k r)h_l(r)dr=f_l^{(L)}, \quad \forall l \geq 0.
\end{equation}
Recall that $f_l^{(L)}=f_{l,m}^{(L)}$ and $h_l(r)=h_{l,m}(r)$.

Equation \eqref{eq53} has many solutions.

Denote by $h_l^\bot(r)$ any function such that
\begin{equation} \label{eq54}
    \int_0^R  r^2 j_l(k r)h_l^\bot(r)dr=0.
\end{equation}
Then the general solution to equations \eqref{eq53} has the form
\begin{equation} \label{eq55}
     h_l(r)=(-i)^{-l-2}f_l^{(L)}\gamma_l j_l (k r)+c_l h_l^\bot(r),\quad 0 
\leq l \leq L,
\end{equation}
where $c_l$ are arbitrarily constants, and
\begin{equation} \label{eq56}
     \gamma_l := \frac{1}{\int_0^R dr r^2 j_l^2(k r)}.
\end{equation}
We have proved the following result.
\begin{Theorem} \label{thm3}
    The function \eqref{eq46}, with $h_l$ defined in \eqref{eq55} for $0 \leq l \leq L$ and $h_l=0$ for
    $l > L$, solves equation \eqref{eq51}. If $L=L(\epsilon)$ is sufficiently large, so that \eqref{eq49} holds,
    then the function
    \begin{equation} \label{eq57}
        h^{(L)}(y)=\sum_{l=0}^L h_l^{(L)}(r)Y_l(y^0), \quad r=|y|, \quad y^0=\frac{y}{r},
    \end{equation}
    with $h_l^{(L)}(r)$ defined in \eqref{eq55}, satisfies inequality \eqref{eq40}.
\end{Theorem}

Let us give a formula for a potential $q$ such that $h=qu$
approximates $h^{(L)}$ with a desired accuracy. This $q$ is a
solution to the inverse scattering problem (IP).

Given $h^{(L)}(y)$ defined in \eqref{eq57}, let us denote it $h(y)$
for simplicity. Using this function, calculate $q(x)$ formula
\eqref{eq16}. If this $q \in L^2(D)$, then the inverse problem (IP)
is solved. There is no guarantee that Im$q \leq 0$.

If formula \eqref{eq16} does not yield an $L^2(D)$ function, then one uses
$h_\delta(x)$ in place of $h(x)$ and, as was proved in Theorem \ref{thm2},
obtains a potential $q_\delta(x)$ which is a solution to (IP).

Other computational methods can be used for finding $h(y)$ given
$f(\beta)$. For example, one can choose a basis $\{\varphi_j\}$ in
$L^2(B)$, $B=D$ is a ball, and look for
\begin{equation} \label{eq58}
    h_n(x)=\sum_{j=1}^n c_j^{(n)} \varphi_j(x),
\end{equation}
where $c_j^{(n)}$ are constants to be found from the minimization problem
\begin{equation} \label{eq59}
    ||f(\beta)-\sum_{j=1}^n c_j^{(n)}g_j(\beta)||_{L^2(S^2)}=min.
\end{equation}
Here
\begin{equation} \label{eq60}
    g_j(\beta):=-\frac{1}{4\pi} \int_B e^{-ik \beta\cdot y} \varphi_j(y)dy.
\end{equation}
A necessary condition for the minimum in \eqref{eq59} is a 
linear algebraic system
for the coefficients $c_j^{(n)}, 1 \leq j \leq n$.


\end{document}